\documentclass[british]{article}
\usepackage[T1]{fontenc}
\usepackage{geometry}
\usepackage{textcomp}
\usepackage{amsmath}
\usepackage{amsthm}
\usepackage{amssymb}
\usepackage{graphicx}

\makeatletter

\providecommand{\tabularnewline}{\\}

\theoremstyle{plain}
\newtheorem{thm}{\protect\theoremname}
\theoremstyle{remark}
\newtheorem*{rem*}{\protect\remarkname}
\theoremstyle{plain}
\newtheorem{lem}{\protect\lemmaname}

\usepackage{tikz}
\usepackage{simpler-wick}

\makeatother

\usepackage{babel}
\usepackage[bibstyle=numeric,citestyle=numeric-comp,backend=bibtex,sorting=none]{biblatex}
\providecommand{\lemmaname}{Lemma}
\providecommand{\remarkname}{Remark}
\providecommand{\theoremname}{Theorem}

\begin{document}
\title{The bosonic Hubbard model on a three dimensional flat band lattice}
\author{Leon Haag-Fank and Andreas Mielke\\
 ~\\
 Institut f\"{u}r Theoretische Physik\\
 University of Heidelberg \\
 Philosophenweg 19\\
 D-69121 Heidelberg, Germany}
\maketitle
\begin{abstract}
The lowest eigenstates of the hopping matrix on the line graph of
a cubic lattice with periodic boundary conditions are highly degenerate,
they form a lowest flat band. Further, these states are localized.
If one considers a repulsive bosonic Hubbard model on this lattice
it is possible to construct exact multi-particle ground states simply
by putting particles in the localized single particle ground states
such that they avoid each other. This can be done up to a certain
critical particle number $N_{c}$. We prove that at this particle
number the ground state entropy is subextensive $\propto N_{c}^{2/3}$.
For lower densities the entropy is extensive. We further show that
the problem is related to the number of 4-cycle decompositions of
the cubic lattice with periodic boundary conditions. 
\end{abstract}

\section{Introduction and Results}

\subsection{Introduction}

Correlated fermions or bosons have been studied successfully using
lattice models like the Hubbard model, independently proposed by Hubbard
\cite{hubbard_electron_1963}, Kanamori \cite{kanamori_electron_1963}
and Gutzwiler \cite{gutzwiller_effect_1963} for fermionic systems
and by Gersch and Knollman for bosonic systems \cite{gersch_quantum_1963}.
It describes a lattice whose lattice sites may be occupied by particles.
A particle can move from one lattice point to a neighboring lattice
point. Multiple particles can occupy the same site, but only at an
expense of energy $U$ per pair of particles. 

If the particles are spin-$\tfrac{1}{2}$ fermions, only pairs of
opposite spin may occupy one lattice site, whereas for bosons, there
is no such restriction.

For the fermionic Hubbard model, many rigorous results are known,
for a review see \cite{mielke_hubbard_2015}. For the bosonic case,
less is known. We are especially interested in flat band systems.
There are different ways to obtains lattices with flat bands. In early
studies of flat band systems line graphs are used, see e.g. \cite{mielke_ferromagnetism_1991}.
Other methods to construct flat bands have been proposed like decorated
lattices, see e.g. \cite{tasaki_ferromagnetism_1992,mielke_ferromagnetism_1993}
or in one or two space dimensions more generic methods \cite{maimaiti_compact_2017,maimaiti_universal_2019,maimaiti_flat-band_2021,maimaiti_non-hermitian_2021,tanaka_extension_2020}.
The latter cannot be used in three dimensions. 

For the bosonic Hubbard model in a flat band system it is possible
to obtain rigorous results below a certain critical filling in one
and two dimensions \cite{motruk_bose-hubbard_2012} as well as some
results on pair formation above that critical density, see e.g. \cite{mielke_pair_2018,fronk_localised_2021}.
For the three-dimensional case, no similar results are known.

The aim of the present paper is to partially fill this gap. We consider
the bosonic Hubbard model on a special lattice with a flat band, namely
on the line graph of the cubic lattice. Single particle eigenstates
in the flat band are localized. It is possible to put particles in
these eigenstates such that they avoid each other, i.e. no lattice
site is occupied by two or more particles. This is possible up to
a certain critical density. We investigate the bosonic Hubbard model
on the line graph of the cubic lattice at that critical density and
show also some results below that density.

\subsection{Line graphs}

In order to define our model, we need some notions and results from
graph theory. A graph $G=(V,E)$ is a pair of a set $V(G)$ of sites
or vertices and a set $E(G)$ of edges. Each edge $e\in E(G)$ is
an unordered pair of vertices and may be denoted as $e=\{x,y\}$,
$x,y\in V(G)$. A graph is bipartite, if its vertex set $V$ is the
union of two disjoint sets $V_{1}$ and $V_{2}$ so that each edge
joins a vertex of $V_{1}$ to a vertex of $V_{2}$. The line graph
of $G$ is as well a graph, it is given by $L(G)=(V(L(G)),E(L(G)))$
with $V(L(G))=E(G)$ and $E(L(G))=\{\{e,e'\}\vert e,e'\in E(G),\vert e\cap e'\vert=1\}$.
For a more intuitive understanding, one can imagine the construction
of the line graph from the original graph in the following way: We
draw a vertex on each edge of a representation of $G$ and connect
two vertices by an edge in $L(G)$ if and only if the edges they are
drawn upon have a common vertex in the original graph. 

\subsection{The model\label{subsec:The-model}}

The Hamiltonian of the bosonic Hubbard model on a line graph $L(G)$
is given by 
\begin{equation}
H=\sum_{\{e,e'\}\in E(L(G))}t_{ee'}b_{e'}^{\dagger}b_{e}+\sum_{e\in V(L(G))}U_{e}n_{e}(n_{e}-1)\label{eq:Hubbard}
\end{equation}
where$b_{e}^{\dagger}$ ($b_{e}$) denote the spinless bosonic creation
(annihilation) operators on the vertex $e\in V(L(G))$, $t_{ee'}$
and $U_{e}>0$ are real parameters. $n_{e}=b_{e}^{\dagger}b_{e}$
is the particle number operator on site $e$. Putting two or more
particles on $e$ costs an energy $\propto U_{e}$. We let $G$ be
a finite part of the cubic lattice with periodic boundary conditions.
The simple cubic lattice with periodic boundary conditions is a cartesian
product of three cycles. Let $C_{L}$ be the cycle graph of length
$L$. $G$ is then defined as the Cartesian product $Q^{3}(L_{1},L_{2},L_{3}):=C_{L_{1}}\square C_{L_{2}}\square C_{L_{3}}$
of three cycles of length $L_{1}$, $L_{2}$, and $L_{3}$. $Q^{3}(L_{1},L_{2},L_{3})$
is the three-dimensional torus with side lengths $L_{1}$, $L_{2}$,
and $L_{3}$. The lattice shall be bipartite, therefore the three
cycles forming the cartesian product must have even length. Since
at each vertex $x\in V(G)$ six edges are connected, $L(G)$ is a
graph of connected $K_{6}$. 

Let $B=(b_{xe})_{x\in V(G),e\in E(G)}$ be the incidence matrix of
$G$ with $b_{xe}=1$, if $x\in e$ and $b_{xe}=0$ otherwise. We
let $T=(t_{ee'})_{e,e'\in V(L(G))}=B^{t}B$. For a connected bipartite
graph one can show that $\dim{\rm kern}B=|E|-|V|+1$. In our case,
$|E|=3|V|$. One can show that for each even cycle of $G$, we can
construct an element of the kernel of $B$. With $G$ being bipartite
we may choose the orientation of each edge such that it points from
$V_{1}(G)$ to $V_{2}(G)$. Let $c$ be a cycle on $G$, i.e. a closed
self-avoiding walk. We choose an orientation of c and define $s_{c}=(s_{e,c})_{e\in V(L(G))}$
with
\begin{equation}
s_{e,c}=\begin{cases}
1, & \text{if }e\text{ belongs to }c\\
 & \text{and has\,the\,same\,orientation\,as }c,\\
-1, & \text{if }e\text{ belongs to }c\\
 & \text{and has\,the\,opposite\,orientation\,as }c,\\
0, & \text{otherwise.}
\end{cases}.
\end{equation}
$s_{c}$ vanishes outside $c$ and has is a sequence of alternating
signs on $c$. It is easy to see that $Bs_{c}=0$. The smallest cycles
on the cubic lattice are cycles of length 4, we call them 4-cycles.
Each 4-cycle corresponds to a face in the lattice. There are $|E(G)|=3|V(G)|$
faces, but $\dim{\rm kern}B=|E|-|V|+1=2|V|+1$. Therefore, the states
corresponding to the 4-cycles of $G$ cannot be linearly independent.
This is a main difference to the two-dimensional case. For a bipartite
plane graph, e.g. a finite part of the square lattice, it can be shown
that the states derived from the faces form a basis, they are complete
and linearly independent.

Let $N_{c}$ be the maximum number of edge-disjoint 4-cycles that
can be put on the lattice. We are dealing with edge-disjoint cycles
meaning that we allow those to have a vertex in common. For $G=Q^{3}(L_{1},L_{2},L_{3})$
with all $L_{i}$ even it has been shown \cite{jeevadoss_decomposition_2015,tapadia_cycle_2019}
that a 4-cycle decomposition exists, i.e. a decomposition of the graph
into disjoint cycles of length 4 such that each edge belongs to one
4-cycle. Therefore, $N_{c}=|E(G)|/4$. We consider the Hamiltonian
\eqref{eq:Hubbard} with $N_{c}$ bosons. We will also discuss partially
the case with $N<N_{c}$ bosons.

For $N=N_{c}=|E|/4$, the 4-cycle decompositions span the many-particle
ground states. The reason is that each particle has to sit in a single
particle ground state. Each single particle ground state corresponds
to an even cycle in $G$. If one particle sits on a cycle with more
than 4 edges, then because of $N=|E|/4$ and the fact that the cycles
have to be edge disjoint because of the local repulsive interaction,
another particle has to sit on a cycle with less than 4 edges. Such
a cycle does not exist if all $L_{i}\geq4$.

Let us mention one additional point: Since $T=B^{t}B$, we can write
the Hamiltonian as a sum of local Hamiltonians $H_{x}$. For $N\leq N_{c}$,
the ground states are ground states of each $H_{x}$ separately. Typically,
this property is called frustration free. Therefore, the model we
are treating is frustration free in this sense. 

The first question we wish to address is how many different 4-cycle
decompositions exist for $G=Q^{3}(L_{1},L_{2},L_{3})$. Each 4-cycle
decomposition yields a multi-particle ground state of $H$ with $N_{c}$
bosons since each boson sits in a single particle ground state and
no site on $L(G)$ contains more than one particle. But since the
single particle states formed out of the 4-cycles are not linearly
independent, these multi-particle ground states are eventually not
linearly independent. This leads us to the second question: What is
the ground state degeneracy for $H$ with $N_{c}$ bosons. The third
questions is: What can be said about the ground state degeneracy for
$N<N_{c}$ bosons?

\subsection{Results}
\begin{thm}
\label{4cd_lower_upper}Let $G=Q^{3}(L_{1},L_{2},L_{3})$ with all
$L_{i}$ even , $4\leq L_{1}\leq L_{2}\leq L_{3}$. Then the number
of different 4-cycle decompositions $\Omega_{4}(G)$ on $G$ is bounded
from below by $C_{1}\exp(c_{1}L_{2}L_{3})$ and bounded from above
by $C_{2}\exp(c_{2}L_{2}L_{3})$ with some positive numbers $c_{i},\,C_{i}$.
\end{thm}
Let $S_{4}=\ln\Omega_{4}(G)$ be the corresponding entropy. Theorem
\ref{4cd_lower_upper} shows that $S_{4}=\Omega(|V(G)|^{2/3})$.
\begin{thm}
\label{gs_Nc}Let $H$ in \eqref{eq:Hubbard} be the bosonic Hubbard
model on $L(G)$ with $G=Q^{3}(L_{1},L_{2},L_{3})$ with all $L_{i}$
even, $4\leq L_{1}\leq L_{2}\leq L_{3}$ and let the particle number
be $N=N_{c}=|E(G)|/4$ . Then the ground state degeneracy is bounded
from below by $C_{1}'\exp(c_{1}'L_{2}L_{3})$ and bounded from above
by $C_{2}'\exp(c_{2}'L_{2}L_{3})$ with some positive numbers $c_{i}',\,C_{i}'$.
\end{thm}
Theorem \ref{gs_Nc} shows that the ground state entropy in this case
is $S_{0}(N_{c})\geq c(|V(G)|^{2/3})$ for some finite constant $c$. 
\begin{rem*}
The upper bound in Theorem 2 follows directly from the upper bound
in Theorem 1 since each 4-cycle decomposition of $G$ yields a multi-particle
ground state of the Hamiltonian on $L(G)$. To obtain a ground state
with $N_{c}$ bosons, all particles must sit on a 4-cycle and any
two 4-cycles must have no edge in common to avoid a repulsion on that
edge. Due to the problem that the single particle ground states formed
using the 4-cycles are not linearly independent, the lower bound in
Theorem 1 does not directly yield a lower bound in Theorem 2.

An entropy $\propto N^{\alpha}$ with $0<\alpha<1$, i.e. a subextensive
entropy typically occurs in systems with some kind of frustration.
Since we are dealing with spinless bosons, it is a priori not clear
why there should be some kind of frustration in this system. To our
knowledge, there is no system of interacting or free lattice bosons
with a subextensive entropy, only some supersymmetric lattice models
with a subextensive entropy are known \cite{huijse_quantum_2009,huijse_supersymmetry_2010}.
We will comment on this point later after the proof of Theorem \ref{4cd_lower_upper}. 
\end{rem*}
Let $\rho=N/|E(G)|$ be the particle density in a sufficiently large
system. Eventually we may take the thermodynamic limit. Then, for
$\rho<\rho_{c}=1/4$, the entropy per lattice site is bounded from
below by $h_{b}(4\rho)/4$ with $h_{b}(\rho)=-\rho\ln\rho-(1-\rho)\ln(1-\rho)$.
This is easy to see. We start with one of the 4-cycle decompositions.
It correponds to a multi-particle state at the density 1/4. Now we
take some bosons out so that $\rho<\rho_{c}.$There are $\binom{N_{c}}{N}$
possibilities to do that which yields the corresponding entropy per
site. The bound is not optimal since we start from one of the many
possible ground states at $\rho_{c}=1/4$. 

\section{Proof of Theorem \ref{4cd_lower_upper}}

We start by looking at the 4-cycle decompositions of $G=Q^{3}(L_{1},L_{2},L_{3})$.
Here and in what follows we always assume that all $L_{i}$ even,
$4\leq L_{1}\leq L_{2}\leq L_{3}$. 
\begin{lem}
\label{Lemma1}$G$ has a 4-cycle decomposition.
\end{lem}
\begin{proof}
The simplest 4-cycle decomposition comes in the form of three sets
of ``interwoven towers'', one set for each dimension. A tower consists
of parallel faces stacked exactly over another, spanning the whole
length in this direction. Such a tower decomposition, as we will call
them, is illustrated in Fig. \ref{fig:tower}. The towers have a periodicity
of two cubes in all directions. It is evident, that tower decompositions
are in fact valid 4-cycle decompositions in lattices of even sizes. 
\end{proof}
\begin{center}
\begin{figure}[bh]
\begin{tabular}{c|c|c}
\includegraphics[width=0.3\linewidth]{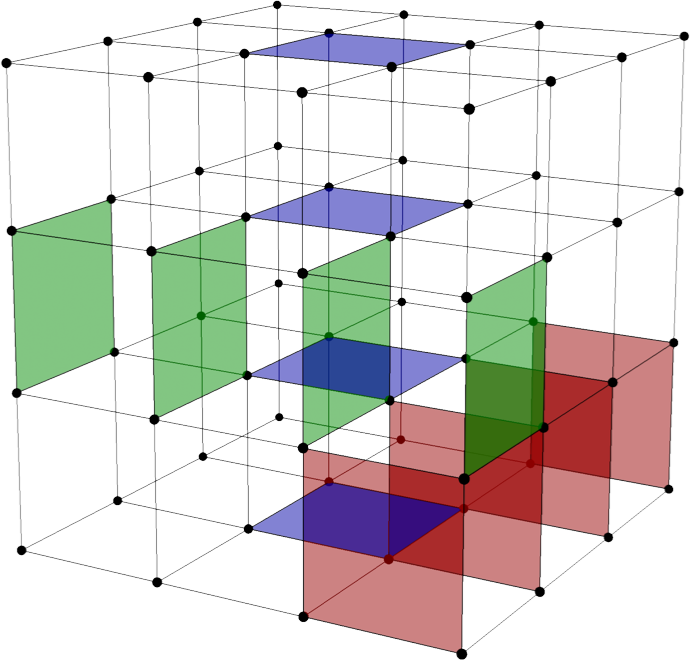} & \includegraphics[width=0.3\linewidth]{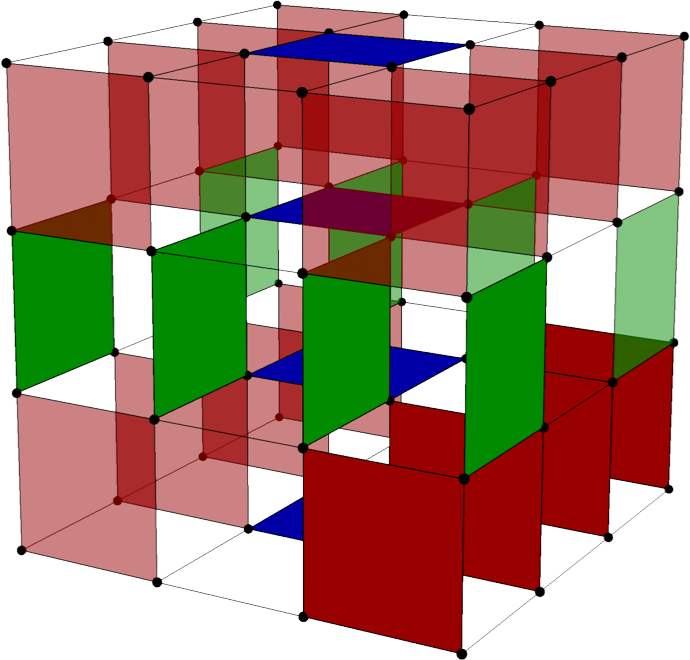} & \includegraphics[width=0.33\linewidth]{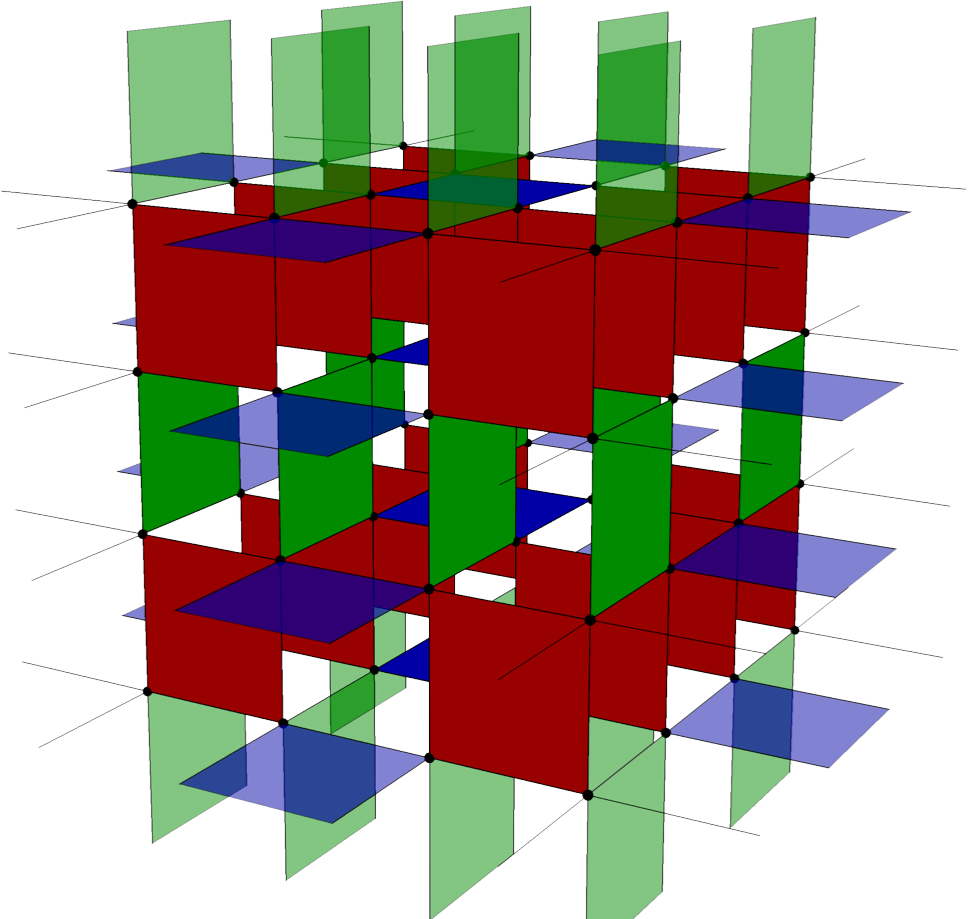}\tabularnewline
\begin{minipage}[t]{0.3\columnwidth}%
One tower for each color%
\end{minipage} & %
\begin{minipage}[t]{0.3\columnwidth}%
All additional towers, that fit in the sub-lattice%
\end{minipage} & %
\begin{minipage}[t]{0.3\columnwidth}%
All towers required to occupy all edges in the sub-lattice%
\end{minipage}\tabularnewline
\end{tabular}\caption{Visualization of a tower decomposition in a 3x3x3 sub-lattice of a
cubic lattice. The three sets of towers are colored in red ($x$-direction),
green ($y$-direction) and blue ($z$-direction).}
\label{fig:tower}
\end{figure}
\par\end{center}

See also \cite{jeevadoss_decomposition_2015,tapadia_cycle_2019} for
different proofs.
\begin{lem}
\label{4cd_lower_bound}The number of different 4-cycle decompositions
of $G$ can be bounded from below by $2^{L_{1}L_{2}/4}+2^{L_{1}L_{3}/4}+2^{L_{2}L_{3}/4}$.
\end{lem}
\begin{proof}
Consider the tower decomposition in Lemma \ref{Lemma1}. It can be
described as towers with faces parallel to the $x-y$-plane and between
these columns with faces perpendicular to the $xy$-plane. If one
rotates one of these columns by $\pi/2$ as depicted in Fig. \ref{fig:The-two-rotations},
one obtains another distinct 4-cycle decomposition. Since all columns
can be rotated independently, one obtains $2^{L_{1}L_{2}/4}$ possible
4-cycle decompositions. This can also be done with columns of faces
perpendicular to the $xz$-plane and with columns of faces perpendicular
to the $yz$-plane. This yields the desired lower bound. 
\end{proof}
\begin{center}
\begin{figure}[bh]
\begin{centering}
\includegraphics[width=0.2\linewidth]{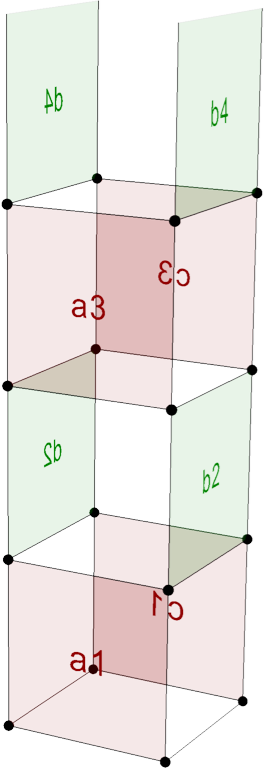} \hspace{0.05\linewidth}
\includegraphics[width=0.2\linewidth]{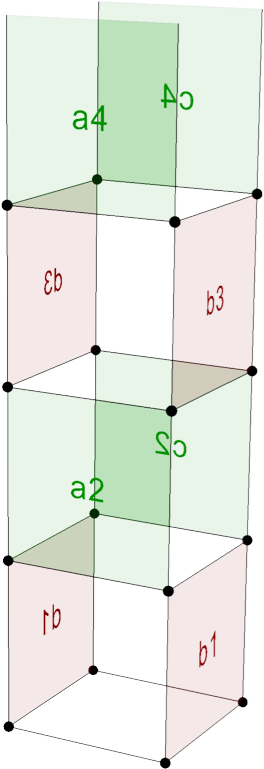} 
\par\end{centering}
\caption{\label{fig:The-two-rotations}The two rotations of a column of a tower
decomposition. The faces are labelled $a_{i},b_{i},c_{i},d_{i}$.
For instance the column in the front right corner in Fig. \ref{fig:tower}
can be rotated in this way around the z -axis by 90\textdegree , yielding
a new, but distinct 4-cycle decomposition.}

\end{figure}
\par\end{center}

Since it is also possible to rotate columns in different directions
as long as these columns do not touch each other, there are more 4-cycle
decompositions of $G$. The lower bound is not optimal.

To obtain an upper bound for the number of 4-cycle decompositions
of $G$, we first prove some general properties of 4-cycle decompositions.
\begin{lem}
\label{Lemma3}In a 4-cycle decomposition of $G$ every vertex is
contained in exactly three 4-cycles, of which one is oriented in $x$-,
one in $y$- and one in $z$-direction, as seen in Fig \ref{fig:Left:-The-center}.
In particular, diagonal arrangements of parallel 4-cycles are not
possible. 
\end{lem}
\begin{center}
\begin{figure}[th]
\begin{centering}
\includegraphics[width=0.3\linewidth]{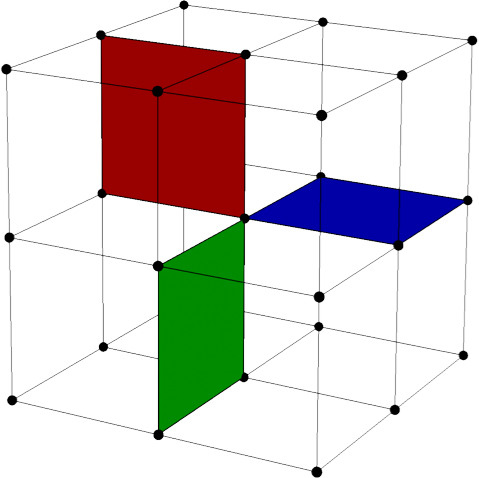}\qquad{}\includegraphics[width=0.3\linewidth]{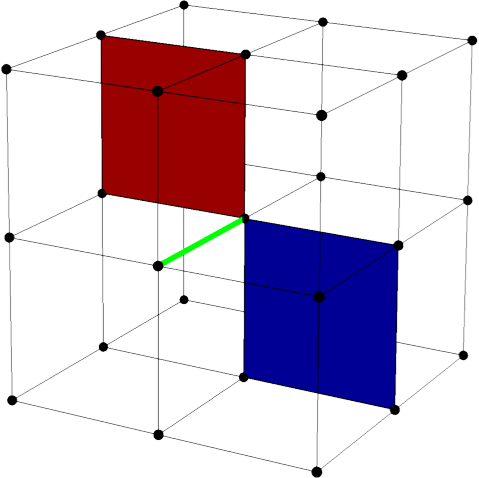}
\par\end{centering}
\caption{\label{fig:Left:-The-center}Left: The center vertex must have exactly
one face per direction attached to it.\protect \\
Right: Two faces diagonally are impossible, since the green edge cannot
be covered.}
\end{figure}
\par\end{center}
\begin{proof}
This is easy to see, by trying out the different possibilities to
cover the 6 edges incident to the vertex with 4-cycles.
\end{proof}
Consider now an arbitrary single plane of squares in the cubic lattice.
Since according to Lemma \ref{Lemma3} two 4-cycles in the plane must
not have a vertex in common, we arrive at a configuration of 4-cycles
in the plane which are hard core with a Moore neighborhood, i.e. a
neighborhood where all eight neighboring 4-cycles on that plane are
forbidden to be occupied. On the other hand, since we aim at a situation
where exactly 1/4 of all 4-cycles are occupied, each vertex in the
plane must exactly belong to one 4-cycle. This means that the 4-cycles
in one plane form a dense packing with a Moore neighborhood. It has
been shown \cite{eloranta_dense_2008}:
\begin{lem}
\label{Lemma:All-4-cycle-2d-configurations}All 4-cycle configurations
of the $xy$-plane with periodic boundary conditions where each vertex
belongs to exactly one 4-cycle and two 4-cycles have no vertex in
common can be constructed from one of the 4 regular configurations
by either shifting an arbitrary number of columns in $x$-direction
or an arbitrary number of columns in $y$-direction by one square,
but not both.
\end{lem}
This means that for one $xy$ plane we obtain 4($2^{L_{1}/2}+2^{L_{2}/2})$
possible configurations. We have $L_{3}$ planes in $xy$ direction,
thus {[}4$(2^{L/2_{1}}+2^{L_{2}/2})]^{L_{3}}$ possible configurations
if we assume that they can be filled independently. Assuming that
in the other two directions we also independently fill the planes,
we obtain
\begin{lem}
\label{4cd_upper_bound}The number of 4-cycle decompositions of $G$
is bounded form above by $[4(2^{L_{1}/2}+2^{L_{2}/2})]^{L_{3}}+4[(2^{L_{1}/2}+2^{L_{3}/2})]^{L_{2}}+[4(2^{L_{2}/2}+2^{L_{3}/2})]^{L_{1}}$. 
\end{lem}
Lemmas \ref{4cd_lower_bound} and \ref{4cd_upper_bound} together
yield Theorem \ref{4cd_lower_upper}.

Let us make a short remark on Lemma \ref{Lemma:All-4-cycle-2d-configurations}.
The dense packing of 4-cycles on a square lattice with Moore neighborhood
can be mapped onto an Ising model on the square lattice with nearest
and next-nearest neighbor interactions and external magnetic field.
The possible configuration of 4-cycles are called row-shifted $2\times2$
states in that case, see e.g. \cite{yin_phase_2009}. From a physics
point of view, the fact that the entropy in the two dimensional case
is proportional to the square root of the number of lattice sites
and to $|E(G)|^{2/3}$ in our case of 4-cycle decompositions on $G$
is related to the frustration of the Ising model due to the competing
nearest and next-nearest neighbor interactions.

\section{Proof of Theorem \ref{gs_Nc}}

As mentioned in the remark to Theorem \ref{gs_Nc}, the upper bound
for the number of different 4-cycle decompositions of $G$ is also
an upper bound for the number of ground states of the bosonic Hubbard
model on $L(G)$ with $N_{c}=|E(G)|/4$ particles. To prove Theorem
\ref{gs_Nc}, we need to prove the lower bound. As a first step we
show 
\begin{lem}
\label{lemma:linear-independence-of}The two possible rotations of
a column of even length $L_{1}$ with periodic boundary as depicted
in Fig. \ref{fig:The-two-rotations} correspond to two linearly independent
multi-particle ground states.
\end{lem}
\begin{proof}
Consider the labeling convention as in Fig. \ref{fig:The-two-rotations}
and abbreviate the creation operators on a 4-cycle $f$ as $b_{f}^{\dagger}$.
The multi-particle ground states of the two rotations are given as
\begin{equation}
|\uparrow\rangle:=\prod_{i=1,{\rm odd}}^{L_{1}}(a_{i}^{\dagger}c_{i}^{\dagger})\prod_{i=2,{\rm even}}^{L_{1}}(b_{i}^{\dagger}d_{i}^{\dagger})|0\rangle,\quad|\downarrow\rangle:=\prod_{i=1,{\rm odd}}^{L_{1}}(b_{i}^{\dagger}d_{i}^{\dagger})\prod_{i=2,{\rm even}}^{L_{1}}(a_{i}^{\dagger}c_{i}^{\dagger})|0\rangle.\label{eq:towerrot}
\end{equation}
 Consider the Gram matrix 
\begin{equation}
M=\begin{pmatrix}\langle\uparrow|\uparrow\rangle & \langle\uparrow|\downarrow\rangle\\
\langle\downarrow|\uparrow\rangle & \langle\downarrow|\downarrow\rangle
\end{pmatrix}.
\end{equation}
We will show, that $\det M\neq0$ and thus $|\uparrow\rangle$ and
$|\downarrow\rangle$ are linearly independent. By \emph{Wick's theorem}
the scalar products in $M$ are each the sum of the full Wick contractions
of creation operators with annihilation operators. Since a contraction
vanishes for faces that have no edge in common, we need not consider
such full contractions. The only relevant contractions are thus 
\begin{equation}
[a_{i},b_{i}^{\dagger}],[b_{i},c_{i}^{\dagger}],[c_{i},d_{i}^{\dagger}],[a_{i},d_{i}^{\dagger}],\quad[f_{i},f_{i\oplus1}^{\dagger}]=-1,\qquad[f_{i},f_{i}^{\dagger}]=4
\end{equation}
along with their swapped variants (e.g. $[b_{i},a_{i}^{\dagger}]$),
where $1\leqslant i\leqslant L_{z}$ and $f_{i}\in\{a_{i},b_{i},c_{i},d_{i}\}$.
Here we write $i\oplus j:=(i+j-1\mod L_{z})+1$, for example $L_{z}\oplus1=1$.

With this we get 
\begin{align}
\langle\downarrow|\downarrow & \rangle=\langle0|\prod_{i=2,{\rm even}}^{L_{1}}(a_{i}c_{i})\prod_{i=1,{\rm odd}}^{L_{1}}(b_{i}d_{i})\prod_{i=2,{\rm even}}^{L_{1}}(a_{i}^{\dagger}c_{i}^{\dagger})\prod_{i=1,{\rm odd}}^{L_{1}}(b_{i}^{\dagger}d_{i}^{\dagger})|0\rangle\\
 & =\prod_{i=2,{\rm even}}^{L_{1}}[a_{i},a_{i}^{\dagger}][c_{i},c_{i}^{\dagger}]\prod_{i=1,{\rm odd}}^{L_{1}}[b_{i},b_{i}^{\dagger}][d_{i},d_{i}^{\dagger}]=4^{2L_{1}}\\
\langle\uparrow|\uparrow\rangle & =\langle0|\prod_{i=1,{\rm odd}}^{L_{1}}(a_{i}c_{i})\prod_{i=2,{\rm even}}^{L_{1}}(b_{i}d_{i})\prod_{i=1,{\rm odd}}^{L_{1}}(a_{i}^{\dagger}c_{i}^{\dagger})\prod_{i=2,{\rm even}}^{L_{1}}(b_{i}^{\dagger}d_{i}^{\dagger})|0\rangle\\
 & =\prod_{i=1,{\rm odd}}^{L_{1}}[a_{i},a_{i}^{\dagger}][c_{i},c_{i}^{\dagger}]\prod_{i=2,{\rm even}}^{L_{1}}[b_{i},b_{i}^{\dagger}][d_{i},d_{i}^{\dagger}]=4^{2L_{z}}
\end{align}
since for every annihilation operator $f_{i}$ in the product, the
only non-vanishing contraction would be with its conjugate. The mixed
products are 
\begin{align}
\langle\downarrow|\uparrow\rangle & =\langle\uparrow|\downarrow\rangle=\langle0|\prod_{i=2,{\rm even}}^{L_{1}}(a_{i}c_{i})\prod_{i=1,{\rm odd}}^{L_{1}}(b_{i}d_{i})\prod_{i=1,{\rm odd}}^{L_{1}}(a_{i}^{\dagger}c_{i}^{\dagger})\prod_{i=2,{\rm even}}^{L_{1}}(b_{i}^{\dagger}d_{i}^{\dagger})|0\rangle
\end{align}
Since every possible contraction is either zero or $-1$, since the
product never contains a creation along with the corresponding annihilation
operator, all non-vanishing full contractions are $(-1)^{2L_{1}}=1$.
We thus only need to count the non-vanishing full contractions. A
trivial strict upper bound is $4^{2L_{1}}$ since every of the $2L_{1}$-many
annihilation operators on the left has at most $4$ creation operators
on the right corresponding to adjacent faces. This upper bound is
strict, i.e. $\langle\uparrow|\downarrow\rangle<4^{2L_{1}}$, since 
\begin{itemize}
\item for $L_{1}=2$ there are only $3$ possible adjacent faces for each
face, yielding a maximum of $3^{2L_{1}}$ such full contractions. 
\item for $L_{1}>2$ at least some of the $4^{2L_{1}}$-many sets of pairs
of annihilation and creation operators must contain one creation/annihilation
operator twice, which thus do not correspond to a valid non-vanishing
full-contraction. A set of such contractions for $L_{1}=4$ is 
\begin{equation}
\langle\downarrow|\uparrow\rangle=\langle0|\wick{a_{2}\c2c_{2}\c1a_{4}c_{4}b_{1}\c4d_{1}\c3b_{3}d_{3}{\c1a_{1}^{\dagger}}c_{2}^{\dagger}{\c3a_{3}^{\dagger}}c_{3}^{\dagger}{\c2b_{2}^{\dagger}}{\c4d_{2}^{\dagger}}b_{4}^{\dagger}d_{4}^{\dagger}}|0\rangle.
\end{equation}
Notice how every contraction is between adjacent faces and the creation
operators correspond exactly to the faces adjacent to $a_{2}$. Thus
there is no possible non-vanishing contraction for $a_{2}$ in this
case. This example still works for $L_{1}>4$ since the additional
operators do not change the argument. 
\end{itemize}
So we get $\det M=\langle\uparrow|\uparrow\rangle\langle\downarrow|\downarrow\rangle-|\langle\uparrow|\downarrow\rangle|^{2}>4^{4L_{1}}-4^{4L_{1}}=0$. 
\end{proof}
\begin{lem}
\label{mpg_lower_bound}The number of different multi-particle ground
states of the bosonic Hubbard model \eqref{eq:Hubbard} on $G$ can
be bounded from below by $2^{L_{2}L_{3}/4}$.
\end{lem}
\begin{proof}
As in the proof of Lemma \ref{4cd_lower_bound}, consider the tower
decomposition in Lemma \ref{Lemma1}. It can be described as towers
with faces parallel to the $yz$-plane and between these columns with
faces perpendicular to the $yz$-plane. If one rotates one of these
columns by $\pi/2$ as depicted in Fig. \ref{fig:The-two-rotations},
one obtains another 4-cycle decomposition. Since all columns can be
rotated independently, one obtains $2^{L_{2}L_{3}/4}$ possible 4-cycle
decompositions. Each of these 4-cycle decomposition yields a bosonic
multi-particle state as described in subsection \ref{subsec:The-model}.
The different columns which we rotate or not have no edge common.
Since as shown in Lemma \ref{lemma:linear-independence-of} two rotated
states in the same column are linearly independent, so are all of
the states with rotated columns as long as all columns have the same
direction, here perpendicular to the $yz$-plane. This yields the
desired lower bound. 
\end{proof}
This together with the upper bound from Theorem \ref{4cd_lower_upper}
proves Theorem \ref{gs_Nc}.

\printbibliography

\end{document}